\documentclass[11pt]{article}

\usepackage{graphicx}
\usepackage{amsfonts}
\usepackage{amsmath}
\usepackage{amsthm}
\usepackage{amssymb}
\usepackage{bm}                     
\usepackage{color}
\usepackage{indentfirst}
\usepackage{multirow}
\usepackage{longtable}
\usepackage{bbm}
\usepackage[all]{xy}
\usepackage{pdflscape}

\usepackage[paper=a4paper,dvips,top=3cm,left=2.4cm,right=2.4cm,
    foot=1cm,bottom=3cm]{geometry}

\allowdisplaybreaks[4]

\begin{document}

\title{A Polynomially Irreducible Functional Basis of Elasticity Tensors}
\author{Zhenyu Ming\footnote{Department of Mathematical Sciences, Tsinghua University, Beijing 100084, China ({\tt mingz17@mails.tsinghua.edu.cn}).}
\and Yannan Chen\footnote{%
    School of Mathematical Sciences, South China Normal University, Guangzhou 510631, China ({\tt ynchen@scnu.edu.cn}).
    This author's work was supported by the National Natural Science Foundation of China (Grant No. 11571178, 11771405).}
\and Liqun Qi\footnote{%
    Department of Applied Mathematics, The Hong Kong Polytechnic University,
    Hung Hom, Kowloon, Hong Kong ({\tt maqilq@polyu.edu.hk}).
    This author's work was partially supported by the Hong Kong Research Grant Council
    (Grant No. PolyU 15302114, 15300715 and 15301716).}
\and Liping Zhang\footnote{Department of Mathematical Sciences, Tsinghua University, Beijing 100084, China ({\tt lipingzhang@tsinghua.edu.cn}).
This author's work was supported by the National Natural Science Foundation of China (Grant No. 11771244).}
}

\date{}
\maketitle

\begin{abstract}

Olive, Kolev and Auffray (2017) proposed a minimal integrity basis of three-dimensional (3D) elasticity tensors, which consists of 297 isotropic invariants, and is also a functional basis.
In this paper, we construct a new functional basis to separate the whole space of 3D elasticity tensors, using Smith's approach. To achieve this goal, we first construct 22 intermediate tensors consisting of 11 second-order symmetrical tensors and 11 scalars via the irreducible decomposition of 3D elasticity tensors. Then, based on Zheng's results, a functional basis consisting of 429 isotropic invariants is obtained by these intermediate tensors. Subsequently, the cardinality is lowered from 429 to 251 after eliminating all the invariants that are zeros or polynomials of the others. Finally, we build a polynomially irreducible functional basis of 3D elasticity tensors, which contains a minimal number of elements compared to existing literature, even if it might not be a minimal functional basis.

\end{abstract}

\newtheorem{Theorem}{Theorem}[section]
\newtheorem{Definition}[Theorem]{Definition}
\newtheorem{Lemma}[Theorem]{Lemma}
\newtheorem{Corollary}[Theorem]{Corollary}
\newtheorem{Proposition}[Theorem]{Proposition}
\newtheorem{Conjecture}[Theorem]{Conjecture}
\newtheorem{Question}[Theorem]{Question}
\newtheorem{Remark}[Theorem]{Remark}

\newcommand{\REAL}{\mathbb{R}}
\newcommand{\COMP}{\mathbb{C}}
\newcommand{\vt}[1]{{\bf #1}}
\newcommand{\aaa}{{\vt{a}}}
\newcommand{\ddd}{{\vt{d}}}
\newcommand{\uu}{{\vt{u}}}
\newcommand{\vv}{{\vt{v}}}
\newcommand{\x}{{\vt{x}}}
\newcommand{\y}{{\vt{y}}}
\newcommand{\z}{{\vt{z}}}
\newcommand{\e}{{\vt{e}}}
\newcommand{\A}{{\bf A}}
\newcommand{\B}{{\bf B}}
\newcommand{\D}{{\bf D}}
\newcommand{\C}{{\bf C}}
\newcommand{\E}{{\bf E}}
\newcommand{\F}{{\bf F}}
\newcommand{\HH}{{\bf H}}
\newcommand{\M}{{\bf M}}
\newcommand{\G}{{\bf G}}
\newcommand{\K}{{\bf K}}
\newcommand{\N}{{\bf N}}
\newcommand{\T}{{\bf T}}
\newcommand{\V}{{\bf V}}
\newcommand{\W}{{\bf W}}
\newcommand{\Q}{{\bf Q}}
\newcommand{\PP}{{\bf P}}
\newcommand{\SH}[1]{{\mathbb{H}^{#1}}}
\newcommand{\OO}{{\mathrm{O}(3)}}
\newcommand{\SO}{{\mathrm{SO}(3)}}
\newcommand{\LCT}{{\bm \epsilon}}
\newcommand{\COV}{\mathbf{Cov}}
\newcommand{\INV}{\mathbf{Inv}}
\newcommand{\BS}[1]{{\mathcal{S}_{#1}}}
\newcommand{\uim}{\mathbf{\Large i}}
\newcommand{\tr}{\mathrm{tr}}

\section{Introduction}
In solid mechanics, the material will deform when external stress is applied on it. The deformation of each material point relates to two second-order symmetric tensors, the strain tensor ${\pmb\varepsilon}$ and the local Cauthy tensor ${\pmb\sigma}$. The linear relationship between ${\pmb\varepsilon}$ and ${\pmb\sigma}$ can be modelled as
$${\pmb\sigma}=\E:{\pmb\varepsilon}.$$
Here, $\E$ is a fourth order elasticity tensor within 21 independent components, satisfying the index form
$$E_{ijkl}=E_{ijlk}=E_{klij}.$$
The theory of elasticity is both an important branch of solid mechanics and the basis of material mechanics, structural mechanics, plastic mechanics and some interdisciplinary subjects. In practice, elastic anisotropy is a very common property of materials, and that of homogeneous materials usually varies with the material direction \cite{Le-81}. Therefore, we study the tensor function representation theory, which is an efficient tool to  describe general consistent invariant forms of the nonlinear constitutive equations and  determine the number and the type of scalar variables involved. In recent years, fruitful works about this topic have been developed \cite{SB-97,CHQZ-17,Ol-14,Ol-17,CZ-18,CZM-18,M-19,Liu-18}.

In the field of mechanical, an integrity basis of a certain type of tensors $\T$ means generating the algebra of invariant polynomials, and a finite set of polynomial isotropic invariants separating the $\rm{O_3}$(three-dimensional orthogonal group)-orbits of $\T$ is called a functional basis. In this paper, the form of a functional basis is limited to polynomials, or they could
be approximated with sufficient accuracy by tensor polynomials
of an arbitrarily high degree \cite{Zh-93}. Therefore, we should call it a polynomial functional basis when necessary. According to the invariant theory, for any finite-dimensional representation $V$, the algebra of polynomial invariants is finitely generated. For this reason, a natural question is ``How to determine a group of invariants to separate tensor space with minimal components?" The answer is clear in terms of vectors (first-order tensor) and second-order (skew-)symmetric tensors. Representations in both complete and irreducible forms were well established by Wang \cite{W-701,W-702,W-7071}, Smith \cite{Sm-71}, Boehler \cite{Boe-77}, and simplified by Zheng \cite{Zh-94}. However, to our best knowledge, it is generally an extremely difficult task for higher order ($\ge3$) tensors. Based on this fact, it is better to first find a series of intermediate tensors which can separate $\rm{O}_3$-orbits with orders no greater than two, instead of searching for the invariants directly from the considered high order tensors. This is the motivation of our paper. Moreover, since the order of elasticity tensor is four (even), its $\rm{O}_3$-orbit and $\rm{SO}_3$-orbit are equivalent.

With regard to the elasticity tensors, Vianello \cite{V-97} proposed an integrity basis with 5 isotropic invariants for plane elasticity tensors. In three dimensions, from a group representation viewpoint, Olive, Kolev and Auffray \cite{OKA-17} obtained a minimal integrity basis comprising 297 isotropic invariants by exploiting the link between the $\rm{SO_3}$-action on irreducible tensors (i.e. symmetric and traceless tensor) and the SL(2,$\mathbb{C}$)-action on the space of binary forms, to separate the whole space of elasticity tensors. Very recently, Desmorat, Auffray, Desmorat, Kolev and Olive \cite{DADKO19}  attempted to separate only generic orbits of 3D elasticity tensors through two sets of invariants made of 19 and 21 polynomials respectively (the latter one is more easier to compute).

In this paper, we succeed in finding a polynomial functional basis to separate the whole space of 3D elasticity tensors and lowering the cardinality to 251. It should be noted that there is no evidence indicating that our basis is a minimal functional basis. We first split an elasticity tensor $\E$ into five parts through harmonic decomposition, namely two scalars $\lambda$ and $\mu$, two second-order irreducible tensors $\D^{(1)}$, $\D^{(2)}$, and one fourth-order irreducible tensor $\A$. After a rather cumbersome numerical experiment, we produce 11 second-order symmetric tensors by the components of $\lambda,\ \mu,\ \D^{(1)},\ \D^{(2)}$ and $\A$. In addition with and 11 invariants (scalars) proposed in \cite{CZM-18,SB-97}, we further prove that these 22 intermediate tensors could determine the $\rm{SO}_3$-orbit of $\E$.
For the credibility and the readability of our proof, we may clarify in advance that the methodology used therein is actually quite fundamental, although the procedure seems like a lot of calculation. The main idea comes from Smith's approach \cite{Sm-71}. According to the harmonic decomposition of $\E$, only $\A$'s order is greater than two, so our attempt is to determine the $\rm{SO}_3$-orbit of $\A$ based on these 22 intermediate tensors.

The principle is to specify the magnitude of each component of $\A$ one by one, which, however, seems unlikely to achieve in many cases. For example, supposing $\PP$ is a second-order symmetric tensor (symmetric matrix) in 3D, taking the form of
$$\PP=P_{ij},\ i,j=1,2,3$$
in a certain coordinate system, obviously $\PP$ has 6 independent components
$$P_{11},\ P_{12},\ P_{13},\ P_{22},\ P_{23},\ P_{33},$$
but the functional bases of $\PP$ being $\tr \PP,\ \tr \PP^2,$ and $\tr \PP^3$ are not enough to determine the magnitudes of six components of $\PP$. However, this problem can be solved if we choose a proper orientation of the coordinate axes (by rotation transformation), making $\PP$ a diagonal matrix:
$$\PP=\text{diag}(P_{11}^{'},\ P_{22}^{'},\ P_{33}^{'})$$
with $P_{11}^{'}\ge P_{22}^{'}\ge P_{33}^{'}$.
 In this sense, $P_{11}^{'},\ P_{22}^{'},$ and $ P_{33}^{'}$ can be calculated by $\tr \PP,\ \tr \PP^2,$ and $\tr \PP^3$. This is the main trick of Smith's approach, which is also adopted in our proof.

 Moreover, the basis we construct is also called a polynomially irreducible functional basis since there is no polynomial relation among these 251 invariants. To confirm it, we use {\sl LinearSolve} function in Mathematica to find all possible polynomial relations therein. More specifically,
 we put the considered invariant (suppose its polynomial degree is $d$) at the ¡°right hand side¡± $b$ of a linear equation. Next, we find all the joints in degree $d$ by multiplying some low degree invariants, which, together with all the $d$-degree invariants (except the considered one), constitute the columns of coefficient matrix $A$ at the ``left hand side'' of the linear equation. Suppose that the number of $A$'s columns is $n$. Then, we generate $m(\ge n)$ groups of rational numbers, with each group containing 19 rational numbers, and assign them to $19(=9+5+5)$ independent elements of $\A,\ \D^1,\ \D^2$ respectively. After the assignment, both $A$ and $b$ are determined. Then we implement {\sl LinearSolve} function to solve the linear equation
$$Ax=b,$$
where  $A$ is an $m*n$ matrix and $b$ is an $m-$dimensional vector. The considered invariant is polynomial of the others, if and only if the linear equation has a solution.

This paper is organized as follows. In Section \ref{2.1}, we first review some basic definitions of  both group theory and tensor function representation theory. Then, we present a minimal functional basis of second-order symmetric tensors by Zheng's results in Section \ref{2.2}. Next, 11 second-order symmetric tensors associated with $\E$ are constructed in Section \ref{Recov}, and it is proved in detail that, together with 11 scalars introduced in \cite{CZM-18,SB-97}, these 22 intermediate tensors can determine the $\rm{SO_3}$-orbit of $\E$. Moreover, a set of 429 isotropic invariants are gained based on the intermediate tensors to form a functional basis of $\E$.
In Section \ref{func-basis}, we further find out and verify all the polynomial relations among these 429 invariants and eliminate the redundant ones (details are shown in the supporting material). As a result, there remain 251 invariants in total, which form a polynomially irreducible functional basis of $\E$. Then, we list them in Table \ref{IsoInv} and compare the size in each degree with the results of Olive, Kolev and Auffray \cite{OKA-17}. In the last section, we draw some concluding remarks.

\section{Preliminaries}
\subsection{Basic definitions}\label{2.1}
We first recall some classical terminologies in tensor function representation theory. Denote $\T$ as an even order tensor represented by $T_{i_{1}\ldots i_{m}}$ under some orthogonal coordinate. Suppose that $\Q\in\rm{SO_3}$ is an rotation matrix. Then an rotation of $\T$ can be expressed as:
\begin{equation*}
  (\Q,\T)\mapsto {\T}':=\Q\ast{\bf T} \text{ with }T'_{i_{1}\ldots i_{m}}=Q_{i_{1}j_{1}}\ldots Q_{i_{m}j_{m}}T_{j_{1}\ldots j_{m}}.
\end{equation*}
The set
$$\{\Q*\T:\ \Q\in\rm{SO_3}\}$$
is called the $\rm{SO_3}$-orbit of $\T$. A polynomial function $f$ defined on components of $\T$ calls a polynomial isotropic invariant of $\T$, if for any $\Q\in\rm{SO_3}$, $f$ satisfies
$$f(\Q*\T)=f(\T).$$
Moreover, the definitions of integrity basis and functional basis are given as below.
\begin{Definition}
Let $\{f_1,f_2,\ldots,f_n\}$ be a finite set of isotropic invariants of  $\bf T$. If any polynomial isotropic invariant of  $\bf T$ is polynomial in $f_1,f_2,\ldots,f_n$, we call the set $\{f_1,f_2,\ldots,f_n\}$ a set of integrity basis of  $\bf T$. In addition, an integrity basis is minimal
if no proper subset of it is an integrity basis.
\end{Definition}
If we relax invariants from polynomials to scalar-valued functions, we get the definitions of functional basis. On the other hand, to reveal the insight of functional basis that it separate the $\rm{SO_3}$-orbits of tensors, an equivalent definition \cite{Ol-14} of functional basis is given.
\begin{Definition}
Let $\{f_1,f_2,\ldots,f_n\}$ be a finite set of polynomial isotropic invariants of  $\bf T$. If
$$f_i({\bf T_1})=f_i({\bf T_2}),\ \text{for all}\ i=1,2,...,n$$
imply ${\bf T_1}=g*{\bf T_2}$ for some $g\in{\rm SO_3}$, we call $\{f_1,f_2,\ldots,f_n\}$ a set of functional basis of  $\bf T$. In addition, a functional basis is minimal
if no proper subset of it is a functional basis.

\end{Definition}
It should be noted that, the size of minimal integrity basis has been proved to be a fixed number. Nevertheless there is no literature confirming this property could be extended to minimal functional basis yet.

\subsection{Functional basis of second-order symmetric tensors}\label{2.2}
Smith \cite{Sm-71} proposed a constructive approach for determining a functional basis of second-order symmetric tensors $\T_{1},\ldots,\T_{N}$, second-order skew-symmetric tensors $\W_{1},\ldots,\W_{M}$ and vectors $\V_{1},\ldots,\V_{P}$. These functional bases were further proved to be minimal by Pennisi and Trovato \cite{PT-87}, and refined by Zheng \cite{Zh-93}. In this paper, we only concentrate on second-order symmetric tensors $\T_{1},\ldots,\T_{N}$, because each intermediate tensor we propose is second-order symmetrical except the scalar ones. Based on Zheng's results, there are 8 different types of isotropic invariants:
\begin{equation}\label{Sec_Ten}
\tr \T_{i},\ \tr \T_{i}^2,\ \tr \T_{i}^3,\ \tr \T_{i}\T_{j},\ \tr \T_{i}^2\T_{j},\ \tr \T_{i}\T_{j}^2,\ \tr \T_{i}^2\T_{j}^2,\ \tr \T_{i}\T_{j}\T_{k},\quad 1\le i<j<k\le N.
\end{equation}

\section{Recovery of an elasticity tensor}\label{Recov}
Let $\SH{m}$ be the space of $m$th-order symmetric and traceless tensors. A symmetric and traceless tensor is also called an irreducible for simplicity. Here, an $m$th-order $n$-dimensional traceless tensor means that
\begin{equation*}
  T_{iii_3\ldots i_m}=0, ~~\forall i_3,\ldots,i_m\in\{1,\cdots,n\}.
\end{equation*}
Now, we factorize the space of elasticity tensors $\mathbb{E}$ into five irreducible subspaces \cite{Zou-01}:
\begin{equation*}
  \mathbb{E} ~~\to~~ \mathbb{R}^1\oplus\mathbb{R}^1\oplus\SH{2}\oplus\SH{2}\oplus\SH{4}.
\end{equation*}
In a more specific form, we have
\begin{equation}\label{decom}
\begin{aligned}
E_{ijkl}=&\{\lambda\delta_{ij}\delta_{kl}\}+\{\mu(\delta_{ik}\delta_{lj}+\delta_{il}\delta_{kj})\}\\
&+\{\delta_{ij}D^{(1)}_{kl}+\delta_{kl}D^{(1)}_{ij}\}+\{\frac{1}{2}(\delta_{ki}D^{(2)}_{jl}+\delta_{kj}D^{(2)}_{il}+\delta_{li}D^{(2)}_{jk}+\delta_{lj}D^{(2)}_{ik})\}+\{A_{ijkl}\}.
\end{aligned}
\end{equation}
It is known that the dimensions of $\SH{2}$ and $\SH{4}$ are 5 and 9 respectively. Hence an elasticity tensor has $21(=1+1+5+5+9)$ independent elements in total. Results in \eqref{Sec_Ten} implies we are able to determine the $\rm{SO}_3$-orbit of a group of second-order symmetric tensors. Hence, the key point of recovering an elasticity tensor is to determine the $\rm{SO}_3$-orbit of the fourth-order irreducible tensor $\A$ via a group of second-order symmetric tensors. As we have mentioned, a valuable thought in Smith's method is to choose a proper orientation of the coordinate axes such that some tensors have better structure. For instance, any second-order symmetric tensor is diagonalizable under some rotation transformation, we could thus set a proper coordinate to make that tensor  diagonal for the simplicity of the discussion in some cases.

In the following part, we mimic the proof in \cite{Sm-71} for the recovery of elasticity tensor. (To help referee(s) check the result, we will supply our Mathematica code.)
The work starts from the decomposition form of $\E$. According to the decomposition \eqref{decom}, we construct 11 second-order symmetric tensors as:
\begin{equation}\label{INTE}
\left\{\begin{aligned}
  &D^{(1)}_{ij},\qquad D^{(2)}_{ij},\qquad B_{ij}:=A_{ik\ell m}A_{jk\ell m}, \qquad
  C_{ij}:=A_{ijk\ell}B_{k\ell}, \qquad
  D_{ij}:=A_{ijk\ell}B_{km}B_{\ell m}, \\
  &F_{ij}:=A_{ijk\ell}D^{(1)}_{k\ell}, \qquad
  H_{ij}:=A_{ipq\ell}A_{jpqm}D^{(1)}_{\ell m}, \qquad
  M_{ij}:=A_{ijk\ell}D^{(1)}_{km}D^{(1)}_{\ell m},\\
    &G_{ij}:=A_{ijk\ell}D^{(2)}_{k\ell}, \qquad
  K_{ij}:=A_{ipq\ell}A_{jpqm}D^{(2)}_{\ell m}, \qquad
  N_{ij}:=A_{ijk\ell}D^{(2)}_{km}D^{(2)}_{\ell m}.
  \end{aligned}\right.
\end{equation}
From the definitions, it can be seen that $\C,\ \D,\ \F,\ \M,\ \G$ and $\N$ are also irreducible tensors.

{\bf First}, if $\D^{(1)}=\D^{(2)}={\bf 0}$, we denote $B^2_{ij}=B_{ik}B_{jk}$ and $P_{ijkl}=A_{ijmn}A_{klmn}$. As a result of \cite{CZM-18,SB-97}, $\A$ could be determined by its functional basis $\{J_2,\ldots,J_{10}\}$, where
\begin{equation*}
\left\{\begin{aligned}
&J_{2}:=A_{ijkl}A_{ijkl},\qquad\quad\quad
J_{3}:=P_{ijkl}A_{ijkl},\qquad \ \
J_{4}:=B_{ij}B_{ij},\\
&J_{5}:=B_{ij}A_{ijkl}B_{kl},\qquad \quad J_{6}:=B_{ij}P_{ijkl}B_{kl},\qquad
J_{7}:=B_{ij}^2A_{ijkl}B_{kl},\\
&J_{8}:=B_{ij}^2P_{ijkl}B_{kl},\qquad\quad J_{9}:=B_{ij}^2A_{ijkl}B_{kl}^2,\quad \quad J_{10}:=B_{ij}^2P_{ijkl}B_{kl}^2.
  \end{aligned}\right.
\end{equation*}

{\bf Second}, we consider the case that $\D^{(1)}$ and $\D^{(2)}$ are not all zero tensors and they are in direct proportion. Without loss of generality, we assume $\D^{(2)}=\rho\D^{(1)}$ for a constant $\rho$. Then, we consider seven tensors defined in \eqref{INTE}: $\D^{(1)},\ \B,\ \C,\ \D,\ \F,\ \HH,\ \M$, and prove that the $\rm{O_3}$-orbit of $\A$ could be determined by these seven tensors.

To make the proof more clearly, we should first give two propositions for the special cases we are to process in the following discussion.
\begin{Proposition}\label{pro1}
Let $\A$ be a fourth-order irreducible tensor, $\alpha,\ \beta\ \text{and}\ \gamma$ be three real numbers, satisfying
\begin{equation*}
\begin{aligned}
 &A_{1112}=A_{1123}=A_{1222}=A_{1223}=A_{2223}=0,\\
 &A_{1122}=\alpha,\ A_{2222}=-2\alpha, A_{1111}=\beta,\  A_{1113}=\gamma.
\end{aligned}
\end{equation*}
Let $\D^{(1)}$ and $\D^{(2)}$ be two diagonal matrixes, $\zeta$ and $\rho$ be two real numbers, satisfying $\D^{(2)}=\rho\D^{(1)}$, $D^{(1)}_{11}=D^{(1)}_{33}=\zeta$ and $D^{(1)}_{22}=-2\zeta$. Suppose that $\B,\ \C,\ \D,\ \F,\ \G,\  \HH,\ \K,\ \M,\ \N$ are defined as \eqref{INTE}. Moreover, we define a fourth-order irreducible tensor $\tilde{\A}$, satisfying
\begin{equation*}
\begin{aligned}
  &\tilde{A}_{1112}=\tilde{A}_{1123}=\tilde{A}_{1222}=\tilde{A}_{1223}=\tilde{A}_{2223}=0,\\
  &\tilde{A}_{1122}=\alpha,\ \tilde{A}_{2222}=-2\alpha,\ \tilde{A}_{1111}=-\frac{3}{4}\alpha,\ \tilde{A}_{1113}=\eta,
\end{aligned}
\end{equation*}
where $$\eta:=\sqrt{\left(\beta+\frac{3}{4}\alpha\right)^2+\gamma^2}.$$
Then, $(\A,\ \D^{(1)},\ \B,\ \C,\ \D,\ \F,\ \G,\  \HH,\ \K,\ \M,\ \N)$ and $(\tilde{\A},\ \D^{(1)},\ \B,\ \C,\ \D,\ \F,\ \G,\  \HH,\ \K,$ $\M,\ \N)$ are in the same $\rm{SO_3}$-orbit.
\end{Proposition}
\begin{proof}
With some calculations, we have
\begin{align*}
  & B_{ij}=\mathrm{diag}(2 \left(3 \alpha ^2+3 \alpha  \beta +2 \beta ^2+2 \gamma ^2\right),\ 10\alpha^2,\ 2 \left(3 \alpha ^2+3 \alpha  \beta +2 \beta ^2+2 \gamma ^2\right)),\\
  & C_{ij}=\mathrm{diag}(2 \alpha  \left(2 \alpha ^2-3 \alpha  \beta -2 \beta ^2-2 \gamma ^2\right),\ -4 \alpha  \left(2 \alpha ^2-3 \alpha  \beta -2 \beta ^2-2 \gamma ^2\right),\\
   &\qquad\qquad\quad 2 \alpha  \left(2 \alpha ^2-3 \alpha  \beta -2 \beta ^2-2 \gamma ^2\right)),\\
  & D_{ij}=\mathrm{diag}(4 \alpha  \left(2 \alpha ^2-3 \alpha  \beta -2 \beta ^2-2 \gamma ^2\right) \left(8 \alpha ^2+3 \alpha  \beta +2 \beta ^2+2 \gamma ^2\right),\\
   &\qquad\qquad\ -8 \alpha  \left(2 \alpha ^2-3 \alpha  \beta -2 \beta ^2-2 \gamma ^2\right) \left(8 \alpha ^2+3 \alpha  \beta +2 \beta ^2+2 \gamma ^2\right),\\
   &\qquad\qquad\quad 4 \alpha  \left(2 \alpha ^2-3 \alpha  \beta -2 \beta ^2-2 \gamma ^2\right) \left(8 \alpha ^2+3 \alpha  \beta +2 \beta ^2+2 \gamma ^2\right)),\\
  & F_{ij}=\mathrm{diag}(-3 \alpha  \zeta,\ 6 \alpha  \zeta,\ -3 \alpha  \zeta),\\
  & G_{ij}=\mathrm{diag}(-3 \alpha\rho  \zeta,\ 6 \alpha \rho \zeta,\ -3 \alpha\rho  \zeta),\\
  & H_{ij}=\mathrm{diag}(2 \zeta  \left(3 \alpha  \beta +2 \beta ^2+2 \gamma ^2\right),\ -8 \alpha ^2 \zeta,\ 2 \zeta  \left(3 \alpha  \beta +2 \beta ^2+2 \gamma ^2\right)),\\
    & K_{ij}=\mathrm{diag}(2 \rho\zeta  \left(3 \alpha  \beta +2 \beta ^2+2 \gamma ^2\right),\ -8 \alpha ^2\rho \zeta,\ 2 \rho\zeta  \left(3 \alpha  \beta +2 \beta ^2+2 \gamma ^2\right)),\\
  & M_{ij}=\mathrm{diag}(3 \alpha  \zeta ^2,\ -6 \alpha  \zeta ^2,\ 3 \alpha  \zeta ^2),\\
   & N_{ij}=\mathrm{diag}(3 \alpha  \rho^2\zeta ^2,\ -6 \alpha  \rho^2\zeta ^2,\ 3 \alpha  \rho^2\zeta ^2).
\end{align*}
By the special structural of $\D^{(1)},\ \B,\ \C,\ \D,\ \F,\ \G,\  \HH,\ \K,\ \M\ \text{and}\ \N$, we further find that these second-order tensors are invariant under rotation transformations in 1-3 plane. If we rotate $\A$ in 1-3 plane of angle $\theta$, anticlockwise, the rotation could be expressed as:
\begin{equation*}
\begin{aligned}
&A_{1112}: 0\to0,\qquad A_{1123}: 0\to0,\qquad A_{1222}: 0\to0,\qquad A_{1223}: 0\to0,\qquad A_{2223}: 0\to0,\\
&A_{1122}: \alpha\to\alpha,\quad A_{2222}:-2\alpha\to-2\alpha,\quad A_{1111}: \beta\to\tilde{\beta}:=(\beta+\frac{3}{4}\alpha)\cos(4\theta)-\gamma\sin(4\theta)-\frac{3}{4}\alpha,\\
&A_{1113}: \gamma\to\tilde{\gamma}:=\gamma\cos(4\theta)+(\beta+\frac{3}{4}\alpha)\sin(4\theta).
\end{aligned}
\end{equation*}
We denote an angle $\phi$, satisfying
$$\cos(\phi):=\frac{\beta+\frac{3}{4}\alpha}{\sqrt{(\beta+\frac{3}{4}\alpha)^2+\gamma^2}}=\frac{\beta+\frac{3}{4}\alpha}{\eta}\quad\text{and}\quad\sin(\phi):=\frac{\gamma}{\sqrt{(\beta+\frac{3}{4}\alpha)^2+\gamma^2}}=\frac{\gamma}{\eta}.$$
Obviously, we have
$$\tilde{\beta}=\eta\cos(\phi+4\theta)-\frac{3}{4}\alpha\quad\text{and}\quad\tilde{\gamma}=\eta\sin(\phi+4\theta).$$
Hence, there exist an angle $\theta$ ($\theta=\frac{\pi-2\phi}{8}$), permitting
$$(\tilde{\beta},\ \tilde{\gamma})=(-\frac{3}{4}\alpha,\ \eta).$$
We thus finish the proof.
\end{proof}

\begin{Proposition}\label{pro2}
Let $\A$ be a fourth-order irreducible tensor, $\alpha,\ \beta,\ \gamma$ be three real numbers, satisfying
\begin{equation*}
\begin{aligned}
  &A_{1113}=A_{1222}=A_{1223}=A_{2223}=0,\\
  &A_{1122}=\alpha,\ A_{2222}=-2\alpha, A_{1111}=-\frac{3}{4}\alpha,\  A_{1112}=\beta,\ A_{1123}=\gamma.
\end{aligned}
\end{equation*}
Suppose $\D^{(1)}$ and $\D^{(2)}$ are two diagonal matrixes, $\zeta$ and $\rho$ are two real numbers, satisfying $\D^{(2)}=\rho\D^{(1)}$, $D^{(1)}_{11}=D^{(1)}_{33}=\zeta$ and $D^{(1)}_{22}=-2\zeta$. Moreover, $\B,\ \C,\ \D,\ \F,\ \G,\  \HH,$ $\K,\ \M,\ \N$ are defined as \eqref{INTE}. Moreover, we define a fourth-order irreducible tensor $\tilde{\A}$, satisfying
\begin{equation*}
\begin{aligned}
  &\tilde{A}_{1113}=\tilde{A}_{1222}=\tilde{A}_{1223}=\tilde{A}_{2223}=0,\\
  &\tilde{A}_{1122}=\alpha,\ \tilde{A}_{2222}=-2\alpha,\ \tilde{A}_{1111}=-\frac{3}{4}\alpha,\ \tilde{A}_{1112}=0,\ \tilde{A}_{1123}=\eta,
\end{aligned}
\end{equation*}
where $$\eta:=\sqrt{\beta^2+\gamma^2}.$$
Then, $(\A,\ \D^{(1)},\ \B,\ \C,\ \D,\ \F,\ \G,\  \HH,\ \K,\ \M,\ \N)$ and $(\tilde{\A},\ \D^{(1)},\ \B,\ \C,\ \D,\ \F,\ \G,\  \HH,\ \K,$ $\M,\ \N)$ are in the same $\rm{SO_3}$-orbit.
\end{Proposition}
\begin{proof}
The proof is similar to Proposition \ref{pro1}. With some calculations, we have
\begin{align*}
  & B_{ij}=\mathrm{diag}\left(\frac{3}{4}(8(\beta^2+\gamma^2)+5\alpha^2\right),\ 2(2(\beta^2+\gamma^2)+5\alpha^2),\ \frac{3}{4}(8(\beta^2+\gamma^2)+5\alpha^2),\\
  & C_{ij}=\mathrm{diag}\left(\frac{1}{4}\alpha(-8(\beta^2+\gamma^2)+25\alpha^2\right),\ -\frac{1}{2}\alpha(-8(\beta^2+\gamma^2)+25\alpha^2),\\ &\qquad\qquad\quad \frac{1}{4}\alpha(-8(\beta^2+\gamma^2)+25\alpha^2)),\\
  & D_{ij}=\mathrm{diag}\left(\frac{5}{16}\alpha(-8(\beta^2+\gamma^2\right)+25\alpha^2)(8(\beta^2+\gamma^2)+11\alpha^2),\\
   &\qquad\qquad-\frac{5}{8}\alpha(-8(\beta^2+\gamma^2)+25\alpha^2)(8(\beta^2+\gamma^2)+11\alpha^2),\\
   &\qquad\qquad\quad \frac{5}{16}\alpha(-8(\beta^2+\gamma^2)+25\alpha^2)(8(\beta^2+\gamma^2)+11\alpha^2)),\\
  & F_{ij}=\mathrm{diag}(-3\alpha\zeta,\ 6\alpha\zeta,\ -3\alpha\zeta),\\
  & G_{ij}=\mathrm{diag}(-3\alpha\rho\zeta,\ 6\alpha\rho\zeta,\ -3\alpha\rho\zeta),\\
  & H_{ij}=\mathrm{diag}\left(-\frac{9}{4}\alpha^2\zeta,\ 4(\beta^2+\gamma^2-2\alpha^2\right)\zeta,\ -\frac{9}{4}\alpha^2\zeta),\\
  & K_{ij}=\mathrm{diag}\left(-\frac{9}{4}\alpha^2\rho\zeta,\ 4(\beta^2+\gamma^2-2\alpha^2\right)\rho\zeta,\ -\frac{9}{4}\alpha^2\rho\zeta),\\
  & M_{ij}=\mathrm{diag}(3\alpha\zeta^2,\ -6\alpha\zeta,\ 3\alpha\zeta^2),\\
  & N_{ij}=\mathrm{diag}(3\alpha\rho^2\zeta^2,\ -6\alpha\rho^2\zeta,\ 3\alpha\rho^2\zeta^2).
\end{align*}
Analogous to the results in Proposition \ref{pro1}, $\D^{(1)},\ \B,\ \C,\ \D,\ \F,\ \G,\  \HH,\ \K,\ \M\ \text{and}\ \N$ are invariant under rotation transformation in 1-3 plane. If we rotate $\A$ in 1-3 plane of angle $\theta$, anticlockwise, the rotation could be presented by:
\begin{equation*}
\begin{aligned}
&A_{1113}: 0\to0,\qquad A_{1222}: 0\to0,\qquad A_{1223}: 0\to0,\qquad A_{2223}: 0\to0,\\
&A_{1122}: \alpha\to\alpha,\qquad A_{2222}:-2\alpha\to-2\alpha,\qquad A_{1111}:-\frac{3}{4}\alpha\to-\frac{3}{4}\alpha,\\
&A_{1112}: \beta\to\tilde{\beta}:=\beta\cos(3\theta)-\gamma\sin(3\theta),\qquad A_{1123}: \gamma\to\tilde{\gamma}:=\beta\sin(3\theta)+\gamma\cos(3\theta).
\end{aligned}
\end{equation*}
We denote an angle $\phi$, satisfying
$$\cos(\phi):=\frac{\beta}{\sqrt{\beta^2+\gamma^2}}=\frac{\beta}{\eta}\quad\text{and}\quad\sin(\phi):=\frac{\gamma}{\sqrt{\beta^2+\gamma^2}}=\frac{\gamma}{\eta}.$$
Obviously, we have
$$\tilde{\beta}=\eta\cos(\phi+3\theta)\quad\text{and}\quad\tilde{\gamma}=\eta\sin(\phi+3\theta).$$
Hence, there exist an angle $\theta$ ($\theta=\frac{\pi-2\phi}{6}$), permitting
$$(\tilde{\beta},\ \tilde{\gamma})=(0,\ \eta).$$
We thus finish the proof.
\end{proof}

\begin{Remark}
We can learn from the above two propositions that under the conditions of Proposition \ref{pro1} (resp. Proposition \ref{pro2}), the $\rm{SO}_3$-orbit of $\A$ can be completely determined by the value of $(A_{1111}+\frac{3}{4}A_{1122})^2+A_{1113}^2$ (resp.  $A_{1112}^2+A_{1123}^2$). The results will be used in subcases (II.2.2.1) and (II.2.2.2) in the following discussion respectively.
\end{Remark}

Now we focus on the main part of the proof. Note that $\B$ is a second-order symmetric tensor, by choosing a proper coordinate system, we could make $\B$ a diagonal tensor, i.e.,
\begin{equation*}
  B_{ij}=\mathrm{diag}(B_{11},B_{22},B_{33}).
\end{equation*}
Furthermore equations $B_{ij}=A_{ik\ell m}A_{jk\ell m}\ (i,j=1,2,3\ \text{and hereinafter})$ could be rewritten as
{\footnotesize
\begin{equation*}
\left\{\begin{aligned}
  & 2 A_{1111}^2+3 A_{1122} A_{1111}+3 A_{1112}^2+2 A_{1113}^2+3 A_{1122}^2+3 A_{1123}^2+2 A_{1222}^2+2
   A_{1223}^2+3 A_{1112} A_{1222}+A_{1113} A_{1223} \\
   &\qquad\qquad\qquad\qquad\qquad\qquad\qquad\qquad\qquad\qquad\qquad\qquad\qquad\qquad\qquad\qquad\qquad\qquad\qquad\qquad\qquad= \frac{1}{2}B_{11}, \\
  & 2 A_{1112}^2+3 A_{1222} A_{1112}+3 A_{1122}^2+2 A_{1123}^2+3 A_{1222}^2+3 A_{1223}^2+2 A_{2222}^2+2
   A_{2223}^2+3 A_{1122} A_{2222}+A_{1123} A_{2223}\\
    &\qquad\qquad\qquad\qquad\qquad\qquad\qquad\qquad\qquad\qquad\qquad\qquad\qquad\qquad\qquad\qquad\qquad\qquad\qquad\qquad\qquad= \frac{1}{2}B_{22}, \\
  & 2 A_{1111}^2+\left(5 A_{1122}+A_{2222}\right) A_{1111}+3 A_{1112}^2+2 A_{1113}^2+5 A_{1122}^2+3
   A_{1123}^2+3 A_{1222}^2+3 A_{1223}^2+2 A_{2222}^2+2 A_{2223}^2 \\&\qquad\qquad\qquad\qquad\qquad\qquad\qquad\qquad +6 A_{1112} A_{1222}+3 A_{1113}
   A_{1223}+5 A_{1122} A_{2222}+3 A_{1123} A_{2223} = \frac{1}{2}B_{33}, \\
  & A_{1122} A_{1123}+2 A_{2222} A_{1123}+3 A_{1113} A_{1222}+A_{1112} \left(4 A_{1113}-3 A_{1223}\right)-2
   A_{1122} A_{2223}-A_{1111} \left(4 A_{1123}+A_{2223}\right)\\
   &\qquad\qquad\qquad\qquad\qquad\qquad\qquad\qquad\qquad\qquad\qquad\qquad\qquad\qquad\qquad\qquad\qquad\qquad\qquad\qquad\qquad = 0, \\
  & -3 A_{1123} A_{1222}+4 A_{2223} A_{1222}+2 A_{1111} A_{1223}+A_{1122} A_{1223}-4 A_{1223}
   A_{2222}-A_{1113} \left(2 A_{1122}+A_{2222}\right)\\
    &\qquad\qquad\qquad\qquad\qquad\qquad\qquad\qquad\qquad\qquad\qquad\qquad\qquad\qquad\qquad\qquad\qquad\qquad+3 A_{1112} A_{2223}= 0, \\
  & 4 A_{1113} A_{1123}+7 A_{1223} A_{1123}+9 A_{1122} A_{1222}+A_{1111} \left(4 A_{1112}+3
   A_{1222}\right)+4 A_{1222} A_{2222}+3 A_{1112} \left(3 A_{1122}+A_{2222}\right) \\&\qquad\qquad\qquad\qquad\qquad\qquad\qquad\qquad\qquad\qquad\qquad\qquad\qquad\qquad\qquad\qquad+A_{1113} A_{2223}+4
   A_{1223} A_{2223} = 0.
\end{aligned}\right.
\end{equation*}
}

Equations $C_{ij}=A_{ijk\ell}B_{k\ell}$ are represented as
\begin{equation}\label{Eqn-C}
\left\{\begin{aligned}
  & A_{1111} \left(B_{11}-B_{33}\right)+A_{1122} \left(B_{22}-B_{33}\right) = C_{11}, \\
  & A_{1112} \left(B_{11}-B_{33}\right)+A_{1222} \left(B_{22}-B_{33}\right) = C_{12}, \\
  & A_{1113} \left(B_{11}-B_{33}\right)+A_{1223} \left(B_{22}-B_{33}\right) = C_{13}, \\
  & A_{1122} \left(B_{11}-B_{33}\right)+A_{2222} \left(B_{22}-B_{33}\right) = C_{22}, \\
  & A_{1123} \left(B_{11}-B_{33}\right)+A_{2223} \left(B_{22}-B_{33}\right) = C_{23}.
\end{aligned}\right.
\end{equation}
\textbf{Case (I). Any two of $B_{11}$, $B_{22}$ and $B_{33}$ are not equal to each other}, i.e., $B_{11}\neq B_{22}$, $B_{22}\neq B_{33}$ and $B_{33}\neq B_{11}$. The symmetric and traceless tensor $D_{ij}=A_{ijk\ell}B_{km}B_{\ell m}$ could be rewritten as
\begin{equation}\label{Eqn-D}
\left\{\begin{aligned}
  & A_{1111} \left(B_{11}^2-B_{33}^2\right)+A_{1122} \left(B_{22}^2-B_{33}^2\right) = D_{11}, \\
  & A_{1112} \left(B_{11}^2-B_{33}^2\right)+A_{1222} \left(B_{22}^2-B_{33}^2\right) = D_{12}, \\
  & A_{1113} \left(B_{11}^2-B_{33}^2\right)+A_{1223} \left(B_{22}^2-B_{33}^2\right) = D_{13}, \\
  & A_{1122} \left(B_{11}^2-B_{33}^2\right)+A_{2222} \left(B_{22}^2-B_{33}^2\right) = D_{22}, \\
  & A_{1123} \left(B_{11}^2-B_{33}^2\right)+A_{2223} \left(B_{22}^2-B_{33}^2\right) = D_{23}.
\end{aligned}\right.
\end{equation}
Combining the first equations of both  \eqref{Eqn-C} and \eqref{Eqn-D}, we get a linear system about $A_{1111}$ and $A_{1122}$:
\begin{equation*}
  \left(\begin{array}{cc}
    B_{11}-B_{33}     & B_{22}-B_{33} \\
    B_{11}^2-B_{33}^2 & B_{22}^2-B_{33}^2 \\
  \end{array}\right)
  \left(\begin{array}{c}
      A_{1111} \\ A_{1122} \\
  \end{array}\right) =
  \left(\begin{array}{c}
      C_{11} \\ D_{11} \\
  \end{array}\right).
\end{equation*}
Since the determinant of its coefficient matrix is
\begin{equation*}
  \det \left(\begin{array}{cc}
    B_{11}-B_{33}     & B_{22}-B_{33} \\
    B_{11}^2-B_{33}^2 & B_{22}^2-B_{33}^2 \\
  \end{array}\right) = -(B_{11}-B_{22})(B_{11}-B_{33})(B_{22}-B_{33}) \neq0,
\end{equation*}
the linear system owns a unique solution
\begin{equation*}
  A_{1111}=\frac{-C_{11}(B_{22}+B_{33})+D_{11}}{(B_{11}-B_{22})(B_{11}-B_{33})} \qquad\text{ and }\qquad
  A_{1122}=\frac{C_{11}(B_{11}+B_{33})-D_{11}}{(B_{11}-B_{22})(B_{22}-B_{33})}.
\end{equation*}
Using a similar approach, we obtain values of
\begin{equation*}
  A_{1112}, \quad A_{1222}, \quad A_{1113}, \quad A_{1223}, \quad A_{2222}, \quad A_{1123}\quad \text{ and }\quad A_{2223}
\end{equation*}
respectively.\\
\textbf{Case (II)}. {\bf All of $B_{11},\ B_{22},\ B_{33}$ are equal}, i.e., $B_{11}=B_{22}=B_{33}$. With simple computations, we find that $\C=\D={\bf 0}$. Moreover, since $\B$ is invariant under any three-dimensional rotation transformation, we could further choose a proper coordinate system to make $\D^{(1)}$ a diagonal tensor and suppose
$$\D^{(1)}:=\mathrm{diag}(d_{1},d_{2},-d_{1}-d_{2}).$$
By identities $F_{ij}=A_{ijk\ell}D^{(1)}_{k\ell}$ and $M_{ij}=A_{ijk\ell}D^{(1)}_{km}D^{(1)}_{\ell m}$, we have
\begin{equation}\label{F}
\left\{\begin{aligned}
  & (2d_{1}+d_{2})A_{1111}+ (d_{1}+2d_{2})A_{1122}=F_{11},\\
  & (2d_{1}+d_{2})A_{1112}+ (d_{1}+2d_{2})A_{1222}=F_{12},\\
  & (2d_{1}+d_{2})A_{1113}+ (d_{1}+2d_{2})A_{1223}=F_{13},\\
  & (2d_{1}+d_{2})A_{1122}+ (d_{1}+2d_{2})A_{2222}=F_{22},\\
  & (2d_{1}+d_{2})A_{1123}+ (d_{1}+2d_{2})A_{2223}=F_{23},
\end{aligned}\right.
\end{equation}
and
\begin{equation}\label{M}
\left\{\begin{aligned}
  & d_{2}(2d_{1}+d_{2})A_{1111}+ d_{1}(d_{1}+2d_{2})A_{1122}=-M_{11},\\
  & d_{2}(2d_{1}+d_{2})A_{1112}+ d_{1}(d_{1}+2d_{2})A_{1222}=-M_{12},\\
  & d_{2}(2d_{1}+d_{2})A_{1113}+ d_{1}(d_{1}+2d_{2})A_{1223}=-M_{13},\\
  & d_{2}(2d_{1}+d_{2})A_{1122}+ d_{1}(d_{1}+2d_{2})A_{2222}=-M_{22},\\
  & d_{2}(2d_{1}+d_{2})A_{1123}+ d_{1}(d_{1}+2d_{2})A_{2223}=-M_{23}.
\end{aligned}\right.
\end{equation}
Similarly, we first consider the determinant of $A_{1111}$ and $A_{1122}$ in the above two linear systems. We find that if
$$(d_{1}-d_{2})(2d_{1}+d_{2})(d_{1}+2d_{2})\neq0,$$
then the determinant is also nonzero. Hence we could determine $A_{1111}$ and $A_{1122}$, and further the rest elements of $\A$ step by step. Otherwise, if
$$(d_{1}-d_{2})(2d_{1}+d_{2})(d_{1}+2d_{2})=0,$$
one of the following three equations should be satisfied:
$$d_{1}=d_{2},\quad 2d_{1}+d_{2}=0\quad\text{and}\quad d_{1}+2d_{2}=0.$$
Without loss of generality, we only need to consider the circumstance of $2d_{1}+d_{2}=0$. Since $\D^{(1)}\neq\bf{0}$, we have that  $d_{1}+2d_{2}\neq0$ and $d_{1}\neq0$. Moreover, equations \eqref{F} and \eqref{M} can be simplified as:
\begin{equation*}
\left\{\begin{aligned}
  & A_{1122}=-\frac{F_{11}}{3d_{1}}=\frac{M_{11}}{3d_{1}^2},\\
  & A_{1222}=-\frac{F_{12}}{3d_{1}}=\frac{M_{12}}{3d_{1}^2},\\
  & A_{1223}=-\frac{F_{13}}{3d_{1}}=\frac{M_{13}}{3d_{1}^2},\\
  & A_{2222}=-\frac{F_{22}}{3d_{1}}=\frac{M_{22}}{3d_{1}^2},\\
  & A_{2223}=-\frac{F_{23}}{3d_{1}}=\frac{M_{23}}{3d_{1}^2}.
\end{aligned}\right.
\end{equation*}
Based on the special form, $\D^{(1)}$ is also invariant under rotation transformation in 1-3 plane. Therefore, we could further choose a proper coordinate system to make $F_{13}=0$. This setting also leads that $A_{1223}=0$ and $M_{13}=0$. Note that the remainder undetermined components of $\A$ are $A_{1111},\ A_{1112},\ A_{1113}$ and $A_{1123}.$ Next, we will take up discussion in more details.\\
(II.1) If $2A_{1122}+A_{2222}\neq0$, from the identity of $H_{13}$, i.e., $H_{13}:=A_{1pq\ell}A_{3pqm}D^{(1)}_{\ell m}$, we have that
$$A_{1113}=-\frac{H_{13}+2A_{1222}A_{2223}}{2d_{1}(2A_{1122}+A_{2222})}.$$
Moreover, from identities of $H_{11}$ and $H_{33}$, we get equations
\begin{equation}\label{A1111}
\left\{\begin{aligned}
  & 4 A_{1111}^2 d_{1}+6 A_{1122} A_{1111} d_{1}+4 A_{1113}^2 d_{1}-2 A_{1222}^2 d_{1}=H_{11},\\
  & 4 A_{1111}^2 d_{1}+(10 A_{1122} +2 A_{2222}) A_{1111} d_{1}+4 A_{1113}^2 d_{1}+4 A_{1122}^2 d_{1}-2 A_{2222}^2 d_{1}-2 A_{2223}^2 d_{1}\\
  &\qquad\qquad\qquad\qquad\qquad\qquad\qquad\qquad\qquad\qquad\qquad\qquad\qquad\qquad\qquad-2 A_{1122} A_{2222} d_{1}=H_{33}.
\end{aligned}\right.
\end{equation}
Considering the coefficients of $A_{1111}^2$ and $A_{1111}$, we find that
\begin{equation*}
\begin{aligned}
(4d_{1})(10 A_{1122} +2 A_{2222})d_{1}-(6A_{1122}d_{1})(4d_{1})
=8d_{1}^2(2 A_{1122} + A_{2222})\neq 0
\end{aligned}
\end{equation*}
is always valid. Thus, we could determine $A_{1111}$ from \eqref{A1111}.
Furthermore, from $B_{13}=0$ and $B_{11}-B_{33}=0$, we know that
\begin{equation}\label{A2223A1222}
\left\{\begin{aligned}
  & 3A_{2223}A_{1112}-3A_{1222}A_{1123}=(2A_{1122}+A_{2222})A_{1113}-4A_{1222}A_{2223},\\
  & -6A_{1222}A_{1112}-6A_{2223}A_{1123}=2(2A_{1122}+A_{2222})A_{1111}+4A_{1122}^2+2A_{1222}^2+10A_{1122}A_{2222}\\
  &\qquad\qquad\qquad\qquad\qquad\qquad\qquad\qquad\qquad\qquad\qquad\qquad\qquad\qquad\qquad+4A_{2222}^2+4A_{2223}^2.
\end{aligned}\right.
\end{equation}
We regard \eqref{A2223A1222} as a linear system of $A_{1112}$ and $A_{1123}$. The determinant is $-18(A_{2223}^2+A_{1222}^2)$.\\
(II.1.1) If$A_{2223}^2+A_{1222}^2\neq0$, we could determine $A_{1112}$ and $A_{1123}$ from \eqref{A2223A1222}.\\
(II.1.2) If$A_{2223}^2+A_{1222}^2=0$, we have $A_{2223}=A_{1222}=0$. Then $B_{13}=-A_{1113}(2A_{1122}+A_{2222})=0$. Since $2A_{1122}+A_{2222}\neq 0$, we have $A_{1113}=0$. With these results, we can simplify the second equation of \eqref{A2223A1222} as
$$\left(2 A_{1122}+A_{2222}\right) \left(A_{1111}+A_{1122}+2 A_{2222}\right)=0.$$
Since $2A_{1122}+A_{2222}\neq 0$, it implies that $$A_{1111}=-A_{1122}-2 A_{2222}.$$ Now the remainder undetermined components of $\A$ are $A_{1112}$ and $A_{1123}$.
Next, from $B_{12}=B_{23}=0$, we obtain
\begin{equation*}
\left\{\begin{aligned}
  & A_{1112}(4A_{1111}+9A_{1122}+3A_{2222})=0,\\
  & -A_{1123}(4A_{1111}-A_{1122}-2A_{2222})=0.
\end{aligned}\right.
\end{equation*}
(II.1.2.1) If $A_{1112}=0$ and $A_{1123}=0$, then $\A$ has been determined.\\
(II.1.2.2) If $A_{1112}=0$ and $A_{1123}\neq0$, we have $A_{1122}=-4A_{1111}+2A_{2222}$. Then, from $B_{23}=0$, we know $$-4A_{1123}(2A_{1111}-A_{2222})=0.$$
Hence, $A_{2222}=2A_{1111}$. In addition, from $B_{22}-B_{33}=0$, we have that
  $$-8A_{1111}^2-2A_{1123}^2=0,$$
 which yields $A_{1123}=0$.\\
(II.1.2.3) If $A_{1112}\neq0$ and $A_{1123}=0$, we have \begin{equation}\label{contra1}
A_{1122}=-\frac{4}{9}A_{1111}-\frac{1}{3}A_{2222}.
 \end{equation}
 From the identity of $H_{12}$, we have
$$\left(\frac{8}{3}A_{1111}-A_{2222}\right)d_{11}A_{1112}=H_{12}.$$
Therefore, if $A_{1111}\neq\frac{3}{8}A_{2222}$, we could determine $A_{1112}$. Otherwise, we have
\begin{equation}\label{contra2}
A_{1111}=\frac{3}{8}A_{2222}.
 \end{equation}
Combining \eqref{contra1} and \eqref{contra2} together, we find that
$A_{1122}=-\frac{1}{2}A_{2222}$. It draws a contradiction to $2A_{1122}+A_{2222}\neq0$.\\
(II.1.2.4) If $A_{1112}\neq0$ and $A_{1123}\neq0$, then we have
\begin{equation*}
\left\{\begin{aligned}
  & 4A_{1111}+9A_{1122}+3A_{2222}=0,\\
  & 4A_{1111}-A_{1122}-2A_{2222}=0.
\end{aligned}\right.
\end{equation*}
 The equations set also implies $2A_{1122}+A_{2222}=0$, which leads to a contradiction.\\
(II.2) If $2A_{1122}+A_{2222}=0$, we consider $A_{1112}$ and $A_{1123}$ first. From $B_{13}=0$ and $B_{11}-B_{33}=0$, we obtain a linear system of $A_{1112}$ and $A_{1123}$:
\begin{equation*}\label{1A2223A1222}
\left\{\begin{aligned}
  & 3A_{2223}A_{1112}-3A_{1222}A_{1123}=-4A_{1222}A_{2223},\\
  & -6A_{1222}A_{1112}-6A_{2223}A_{1123}=2A_{1222}^2+4A_{2223}^2.
\end{aligned}\right.
\end{equation*}
(II.2.1) If the determinant of this linear system $A_{2223}^2+A_{1222}^2\neq0$, we could immediately solve this linear system and determine $A_{1112}$ and $A_{1123}$ by $A_{1222}$ and $A_{2223}$.
The expressions are
\begin{equation}\label{1A1112A1123}
\left\{\begin{aligned}
A_{1112}=&-\frac{A_{1222} \left(A_{1222}^2+6 A_{2223}^2\right)}{3
   \left(A_{1222}^2+A_{2223}^2\right)},\\
A_{1123}=&\frac{A_{2223} \left(3 A_{1222}^2-2 A_{2223}^2\right)}{3
   \left(A_{1222}^2+A_{2223}^2\right)}.
\end{aligned}\right.
\end{equation}
Next, we focus on $A_{1111}$ and $A_{1113}$. Combining \eqref{1A1112A1123} and $B_{12}=B_{23}=0$, a linear system of $A_{1111}$ and $A_{1113}$ is obtained by
\begin{equation*}
\left\{\begin{aligned}
&A_{1111} (4A_{1112}+3 A_{1222})+A_{1113}(4A_{1123}+A_{2223})+3 A_{1122} A_{1112}+A_{1122} A_{1222} =0,\\
&-A_{1111} (4A_{1123}+A_{2223})+A_{1113}(4A_{1112}+3 A_{1222})-3 A_{1122} A_{1123}-2A_{1122} A_{2223} =0.
\end{aligned}\right.
\end{equation*}
We also consider the determinant of the above linear system. If $(4A_{1112}+3 A_{1222})^2+(4A_{1123}+A_{2223})^2\neq0$, we could immediate determine $A_{1111}$ and $A_{1113}$. Otherwise, we have the relations
\begin{equation}\label{2A}
A_{1222}=-\frac{4}{3}A_{1112}\qquad\text{and}\qquad A_{2223}=-4A_{1123}.
\end{equation}
From $B_{12}=B_{13}=B_{23}=0$, we have
\begin{equation*}
\left\{\begin{aligned}
&A_{1112}A_{1122}=0,\\
&A_{1112}A_{1123}=0,\\
&A_{1122}A_{1123}=0.
\end{aligned}\right.
\end{equation*}
(II.2.1.1) If $A_{1112}=0$, then $B_{11}-B_{33}=-40A_{1123}^2=0$, which further yields $A_{1123}=0$. Then from \eqref{2A}, we have $A_{1222}=A_{2223}=0$. This leads to a contradiction to $A_{2223}^2+A_{1222}^2\neq0$.\\
(II.2.1.2) If $A_{1112}\neq0$, then $A_{1122}=A_{1123}=0$. In this case, we have $B_{11}-B_{33}=\frac{40}{9}A_{1112}^2=0$, which yields $A_{1112}=0$. This draws a contradiction to $A_{1112}\neq0$.\\
(II.2.2) If $A_{2223}^2+A_{1222}^2=0$, we have that $A_{2223}=A_{1222}=0$. From $B_{12}=B_{23}=0$, a linear system of $A_{1111}$ and $A_{1113}$ are obtained by
\begin{equation}\label{1A1111A1113}
\left\{\begin{aligned}
&4A_{1111}A_{1112}+4A_{1113}A_{1123}+3 A_{1122} A_{1112}=0,\\
&-4A_{1111}A_{1123}+4A_{1113}A_{1112}-3 A_{1122} A_{1123}=0.
\end{aligned}\right.
\end{equation}
The determinant of this linear system is $A_{1112}^2+A_{1123}^2$.\\
(II.2.2.1) If $A_{1112}^2+A_{1123}^2=0$, then we have $A_{1112}=A_{1123}=0$. Therefore, $\A$ and $\D^{(1)}$ satisfy the conditions of Proposition \ref{pro1}. In addition, from identity of $B_{11}$, we have
\begin{equation*}
(A_{1111}+\frac{3}{4}A_{1122})^2+A_{1113}^2=\frac{1}{4}B_{11}-\frac{15}{16}A_{1122}^2.
\end{equation*}
Hence, by Proposition \ref{pro1}, we could choose
$$A_{1111}=-\frac{3}{4}A_{1122},\quad \text{and}\quad A_{1113}=\sqrt{\frac{1}{4}B_{11}-\frac{15}{16}A_{1122}^2}.$$\\
(II.2.2.2) If $A_{1112}^2+A_{1123}^2\neq0$, from equations \eqref{1A1111A1113}, we have
\begin{equation*}
\left\{\begin{aligned}
&A_{1111}=-\frac{3}{4}A_{1122},\\
&A_{1113}=0.
\end{aligned}\right.
\end{equation*}
Therefore, $\A$ and $\D^{(1)}$ satisfy the conditions of Proposition \ref{pro2}. Whereafter, from $B_{11}-B_{22}=0$, we know $$A_{1112}^2+A_{1123}^2=\frac{25}{8}A_{1122}^2.$$
According to Proposition \ref{pro2}, we could choose
$$A_{1112}=0,\quad \text{and}\quad A_{1123}=\frac{5\sqrt{2}}{4}|A_{1122}|.$$

In conclusion, we could always determine either each element of $\A$, or the $\rm{O}_3$-orbit of $\A$ by $\B,\ \C,\ \F,\ \HH$ and $\M$ in this subcase.\\
\textbf{Case (III).} {\bf Two of $B_{11},\ B_{22},\ B_{33}$ are equal, but they are not equal to the third one}. Without loss of generality, we assume $B_{11}=B_{33}\neq B_{22}$. In this case, we denote $$\D^{(1)}:=(d_{ij}),\ i,j=1,2,3.$$ From \eqref{Eqn-C}, we immediately have
\begin{equation*}
\left\{\begin{aligned}
  & A_{1122}  = \frac{C_{11}}{B _{22}-B _{33}}, \\
  &  A_{1222}  = \frac{C_{12}}{B _{22}-B _{33}}, \\
  &  A_{1223}  = \frac{C_{13}}{B _{22}-B _{33}}, \\
  &  A_{2222}  = \frac{C_{22}}{B _{22}-B _{33}}, \\
  &  A_{2223}  = \frac{C_{23}}{B _{22}-B _{33}}.
\end{aligned}\right.
\end{equation*}
The remainder undetermined components of $\A$ are $A_{1111},\ A_{1112},\ A_{1113}$ and $A_{1123}$. Since $\B$ is still invariant under rotation transformation in 1-3 plane, we could further choose a proper coordinate system to make $d_{13}=0$. From $F_{ij}=A_{ijk\ell}D^{(1)}_{k\ell}$, we know that
\begin{equation}\label{d}
\left\{\begin{aligned}
   &(2d_{11}+d_{22})A_{1111}+2d_{12}A_{1112}+2d_{23}A_{1123}=F_{11}-A_{1122}(d_{11}+2d_{22}),\\
   & (2d_{11}+d_{22})A_{1112}  = F_{12}-A_{1222}(d_{11}+2d_{22})-2A_{1122}d_{12}-2A_{1223}d_{23}, \\
    &-2d_{23}A_{1112}+(2d_{11}+d_{22})A_{1113}+2d_{12}A_{1123}=F_{13}-A_{1223}(d_{11}+2d_{22})+2A_{1222}d_{23},\\
& (2d_{11}+d_{22})A_{1123}=F_{23}-A_{2223}(d_{11}+2d_{22})-A_{1223}d_{12}+2A_{1122}d_{23}+2A_{2222}d_{23}.
\end{aligned}\right.
\end{equation}
(III.1) If $d_{22}\neq-2d_{11}$, then, from the second and the fourth equations of \eqref{d}, we could determine $A_{1112}$ and $A_{1123}$.
Furthermore, from the first and the third equations of \eqref{d}, we could determine $A_{1111}$ and $A_{1113}$.
Thus each element of $\A$ is known.\\
(III.2) If $d_{22}=-2d_{11}$, then $d_{11}\neq0$. Also from $\eqref{d}$, we have
\begin{equation*}
\left\{\begin{aligned}
  & 2d_{12}A_{1112}+2d_{23}A_{1123}=F_{11}+3d_{11}A_{1122},\\
  &  -2d_{23}A_{1112}+2d_{12}A_{1123}=F_{13}+3d_{11}A_{1223}+2A_{1222}d_{23}.
\end{aligned}\right.
\end{equation*}
(III.2.1) If $d_{12}^2+d_{23}^2\neq0$, we could immediately determine $A_{1112}$ and $A_{1123}$, since $A_{1122}$ and $A_{1222}$ are known.
Then, from $M_{ij}=A_{ijk\ell}D^{(1)}_{km}D^{(1)}_{\ell m}$, we obtain two linear equations of $A_{1111}$ and $A_{1113}$ as below.
{\small
\begin{equation*}
\left\{\begin{aligned}
   &A_{1111}
   (d_{12}^2-d_{23}^2)+2 A_{1113} d_{12} d_{23}+3 A_{1122} d_{11}^2-2 A_{1112} d_{12} d_{11}-2 A_{1123} d_{23} d_{11}+A_{1122} d_{12}^2 =M_{11},\\
   &-2 A_{1111}d_{12} d_{23}+A_{1113}(d_{12}^2-d_{23}^2)+3 A_{1223} d_{11}^2+2 A_{1112} d_{23} d_{11}+2 A_{1222}
   d_{23} d_{11}+A_{1223} d_{12}^2 \\
   &\qquad\qquad\qquad\qquad\qquad\qquad\qquad\qquad\qquad\qquad\qquad\qquad\qquad-2 A_{1123} d_{12} d_{11}-2 A_{1122} d_{12} d_{23}=M_{13}.
\end{aligned}\right.
\end{equation*}
}
The determinant of this linear system is $(d_{12}^2-d_{23}^2)^2+4 d_{12}^2 d_{23}^2$.\\
(III.2.1.1) If $(d_{12}^2-d_{23}^2)^2+4 d_{12}^2 d_{23}^2\neq0$, we could immediate determine $A_{1111}$ and $A_{1113}$, hence each element of $\A$ is obtained.\\
(III.2.1.2) If $(d_{12}^2-d_{23}^2)^2+4 d_{12}^2 d_{23}^2=0$. Then we have $d_{12}=d_{23}=0$, which leads to a contradiction to $d_{12}^2+d_{23}^2\neq0$.\\
(III.2.2) If $d_{12}^2+d_{23}^2=0$, then $d_{12}=d_{23}=0$. Thus, both $\B$ and $\D$ are invariant under rotation transformations in 1-3 plane. For this reason, we can choose a proper coordinate system to make $C_{13}=0$, and therefore, $A_{1223}=0$. Next, from equations $H_{ij}=A_{ipq\ell}A_{jpqm}D^{(1)}_{\ell m}$, we obtain a quadratic equation set of $A_{1111}$ and $A_{1113}$.
\begin{equation}\label{eq1}
\left\{\begin{aligned}
&A_{1113}(-2 A_{1122}-A_{2222}) d_{11}-2 A_{1222} A_{2223} d_{11}=H_{13},\\
&4 A_{1111}^2 d_{11}+6 A_{1122} A_{1111} d_{11}+4 A_{1113}^2 d_{11}-2 A_{1222}^2
   d_{11}=H_{11},\\
&4 A_{1111}^2 d_{11}+(10 A_{1122}+2A_{2222})A_{1111} d_{11}+4
  A_{1113}^2 d_{11}+4 A_{1122}^2 d_{11}-2 A_{2222}^2 d_{11}-2 A_{2223}^2 d_{11}\\
&\qquad\qquad\qquad\qquad\qquad\qquad\qquad\qquad\qquad\qquad\qquad\qquad\qquad\qquad-2A_{1122} A_{2222} d_{11}=H_{33}.
\end{aligned}\right.
\end{equation}
(III.2.2.1) If $A_{2222}\neq-2A_{1122}$, from the first equation of \eqref{eq1}, we could directly determine $A_{1113}$. Then, we consider the element $A_{1111}$. From the second and the third equations of \eqref{eq1}, we know that
\begin{equation*}
\begin{aligned}
&2d_{11}(2A_{1122}+A_{2222})A_{1111}\\
&\qquad\qquad=H_{33}-H_{11}+2d_{11}(-2A_{1122}^2+A_{1122}A_{2222}-A_{1222}^2+A_{2222}^2+A_{2223}^2+A_{2223}^2).
\end{aligned}
\end{equation*}
Together with $d_{11}\neq0$, we could obtain $A_{1111}$. Next, $B_{11}-B_{33}=B_{13}=0$ leads to a linear system of $A_{1112}$ and $A_{1123}$ as:
\begin{equation*}
\left\{\begin{aligned}
&3 A_{1112}A_{1222}+3 A_{1123} A_{2223}+2 A_{1122}^2+2 A_{1111} A_{1122}+5 A_{2222} A_{1122}+A_{1222}^2+2 A_{2222}^2\\
&\qquad\qquad\qquad\qquad\qquad\qquad\qquad\qquad\qquad\qquad\qquad\qquad\qquad+2
   A_{2223}^2+A_{1111} A_{2222}=0,\\
&3 A_{1112} A_{2223}-3 A_{1123} A_{1222}-2 A_{1113} A_{1122}-A_{1113} A_{2222}+4
   A_{1222} A_{2223}=0.
\end{aligned}\right.
\end{equation*}
 We only need to consider the case that the determinant of this linear system is zero, which equals to $A_{1222}^2+A_{2223}^2=0$. Then, $A_{1222}=A_{2223}=0$. From $B_{11}-B_{33}=B_{13}=0$, we have
\begin{equation*}
\left\{\begin{aligned}
&\left(2 A_{1122}+A_{2222}\right) \left(A_{1111}+A_{1122}+2 A_{2222}\right)=0,\\
&A_{1113} \left(2 A_{1122}+A_{2222}\right)=0.
\end{aligned}\right.
\end{equation*}
Since $A_{2222}+2A_{1122}\neq0$, we have $A_{1113}=0$ and $A_{1111}=-A_{1122}-2A_{2222}$. From $B_{12}=B_{23}=0$, we know that
\begin{equation*}
\left\{\begin{aligned}
&A_{1112} \left(A_{1122}-A_{2222}\right)=0,\\
&A_{1123} \left(A_{1122}+2A_{2222}\right)=0.
\end{aligned}\right.
\end{equation*}
(III.2.2.1.1) If $A_{1112}=A_{1123}=0$, $\A$ has been determined.\\
(III.2.2.1.2) If $A_{1112}=0$ and $A_{1123}\neq0$, we have $A_{2222}=-\frac{1}{2}A_{1122}$. Then, from identity of $H_{13}$, we know that $$H_{13}=-\frac{9}{2}A_{1122}A_{1123}.$$
Since $A_{2222}+2A_{1122}\neq0$ and $A_{2222}=-\frac{1}{2}A_{1122}$, we find that $A_{1122}\neq0,$ hence $A_{1123}$ could be determined.\\
(III.2.2.1.3) If $A_{1112}\neq0$ and $A_{1123}=0$, we have $A_{2222}=A_{1122}$. Then, from identity of $H_{12}$, we get $$H_{12}=-9A_{1112}A_{1122}d_{11}.$$
For the same reason, $A_{1122}\neq0.$ so we could determine $A_{1112}$.\\
(III.2.2.1.4) If $A_{1112}\neq0$ and $A_{1123}\neq0$, we have $A_{1122}-A_{2222}=0$ and $A_{1122}+2A_{2222}=0$, then, $$A_{1122}=A_{2222}=0,$$ which draws a contradiction to $A_{2222}+2A_{1122}\neq0$.\\
(III.2.2.2) If $A_{2222}=-2A_{1122}$. The rest of the proof in this subcase is analogous to that of (II.2), hence, for simplicity, we omit the discussion here.

{\bf Third}, if both $\bf{D^{(1)}}$ and $\bf{D^{(2)}}$ are not zero-tensors, and they are not in direct proportion either.
It should be noted that, for a subcase in this part, if we finally determine all the elements (instead of the $\rm{SO_3}$-orbit) of $\A$, the proof will be quite similar to those in the {\bf Second} part, due to technically minor changes. For this reason, we only need to start from the subcases (II.2.2.1), (II.2.2.2) and (III.2.2.2) in the {\bf Second} part. To avoid confusion, we denote them as (II$^{*}$.2.2.1), (II$^{*}$.2.2.2) and (III$^{*}$.2.2.2) here respectively. Moreover, we suppose that $$\D^{(2)}:=(\hat{d}_{ij}),\ i,j=1,2,3.$$
\\
(II$^{*}$.2.2.1) Now we have $A_{1223}=A_{2223}=A_{1222}=A_{1112}=A_{1123}=0$ and $A_{2222}=-2A_{1122}$. Our goal is to determine $A_{1111}$ and $A_{1113}$. From $G_{ij}=A_{ijk\ell}D^{(2)}_{k\ell}$, we obtain a linear system of $A_{1111}$ and $A_{1113}$:
\begin{equation*}
\left\{\begin{aligned}
& A_{1111} \left(2\hat{d}_{11}+ \hat{d}_{22}\right)+2 A_{1113} \hat{d}_{13}+A_{1122} \hat{d}_{11}+2 A_{1122} \hat{d}_{22}=G_{11},\\
& A_{1113}(2\hat{d}_{11}+\hat{d}_{22})-2 A_{1111} \hat{d}_{13}-2A_{1122} \hat{d}_{13}=G_{13}.
\end{aligned}\right.
\end{equation*}
The determinant is $(2\hat{d}_{11}+\hat{d}_{22})^2+\hat{d}_{13}^2$.\\
(II$^{*}$.2.2.1.1) If $(2\hat{d}_{11}+\hat{d}_{22})^2+\hat{d}_{13}^2\neq0$, we could determine $A_{1111}$ and $A_{1113}$.\\
(II$^{*}$.2.2.1.2) If $(2\hat{d}_{11}+\hat{d}_{22})^2+\hat{d}_{13}^2=0$, we have $\hat{d}_{22}=-2\hat{d}_{11}$ and $\hat{d}_{13}=0$.
From $N_{ij}=A_{ijk\ell}D^{(2)}_{km}D^{(2)}_{\ell m}$, we further have that
\begin{equation*}
\left\{\begin{aligned}
& A_{1111} (\hat{d}_{12}^2-\hat{d}_{23}^2)+2 A_{1113} \hat{d}_{12} \hat{d}_{23}+3 A_{1122} \hat{d}_{11}^2+A_{1122} \hat{d}_{12}^2=N_{11},\\
& -2 A_{1111} \hat{d}_{23} \hat{d}_{12}+A_{1113} (\hat{d}_{12}^2-\hat{d}_{23}^2)-2 A_{1122} \hat{d}_{23} \hat{d}_{12}=N_{13}.
\end{aligned}\right.
\end{equation*}
If $(\hat{d}_{12}^2-\hat{d}_{23}^2)^2+\hat{d}_{12}^2 \hat{d}_{23}^2\neq0$, we could determine $A_{1111}$ and $A_{1113}$. If not, we have $\hat{d}_{12}=\hat{d}_{23}=0$.
Thus, $\D^{(1)}$ and $\D^{(2)}$ are in direct proportion, which leads to a contradiction to the precondition.\\
(II$^{*}$.2.2.2) Now we have $A_{1223}=A_{2223}=A_{1222}=A_{1113}=0$,  $A_{2222}=-2A_{1122}$ and $A_{1111}=-\frac{3}{4}A_{1122}$ and we need to determine $A_{1112}$ and $A_{1123}$. From $G_{ij}=A_{ijk\ell}D^{(2)}_{k\ell}$, we have a linear system of $A_{1112}$ and $A_{1123}$:
\begin{equation*}
\left\{\begin{aligned}
& A_{1112}(2\hat{d}_{11}+\hat{d}_{22})+2 A_{1123} \hat{d}_{13}+2 A_{1122} \hat{d}_{12}=G_{12},\\
& -2 A_{1112} \hat{d}_{13}+A_{1123} (2\hat{d}_{11}+\hat{d}_{22})+2 A_{1122} \hat{d}_{23}=G_{23}.
\end{aligned}\right.
\end{equation*}
The determinant is $(2\hat{d}_{11}+\hat{d}_{22})^2+\hat{d}_{13}^2$.\\
(II$^{*}$.2.2.2.1) If $(2\hat{d}_{11}+\hat{d}_{22})^2+\hat{d}_{13}^2\neq0$, we could determine $A_{1112}$ and $A_{1123}$.\\
(II$^{*}$.2.2.2.2) If $(2\hat{d}_{11}+\hat{d}_{22})^2+\hat{d}_{13}^2=0$, we have $\hat{d}_{13}=0$ and $\hat{d}_{22}=-2\hat{d}_{11}$. Furthermore, from $N_{ij}=A_{ijk\ell}D^{(2)}_{km}D^{(2)}_{\ell m}$, we know that
\begin{equation*}
\left\{\begin{aligned}
&  A_{1112} (\hat{d}_{12}^2-\hat{d}_{23}^2)+2 A_{1123} \hat{d}_{23} \hat{d}_{12}-2 A_{1122} \hat{d}_{11} \hat{d}_{12}=N_{12},\\
& -2 A_{1112} \hat{d}_{23} \hat{d}_{12}+A_{1123} (\hat{d}_{12}^2-\hat{d}_{23}^2)-2 A_{1122} \hat{d}_{11} \hat{d}_{23} =N_{23}.
\end{aligned}\right.
\end{equation*}
This is a linear system of $A_{1112}$ and $A_{1123}$ and we also consider its determinant. If $(\hat{d}_{12}^2-\hat{d}_{23}^2)^2+\hat{d}_{12}^2\hat{d}_{23}^2\neq0$, we could determine $A_{1112}$ and $A_{1123}$. If not, we have $\hat{d}_{12}=\hat{d}_{23}=0$.
Thus, $\D^{(1)}$ and $\D^{(2)}$ are in direct proportion, which also leads to a contradiction to the precondition.\\
(III$^{*}$.2.2.2) As what we have mentioned in (III.2.2.2), the proof of this subcase can be mimicked by the previous parts, therefore we omit the discussion here.

In sum, we establish the following theorem.
\begin{Theorem}\label{Main_Thm}
The $\rm{SO_3}$-orbit of the elasticity tensor $\E$ can be determined by 11 second-order symmetrical tensors $\D^{(1)},\ \D^{(2)},\ \B,\ \C,\ \D,\ \F,\ \G,\ \HH,\ \K,\ \M,\ \N$ and 11 scalars $\lambda,\ \mu,\ J_{2},\ldots,J_{10}$.
\end{Theorem}

\section{A Polynomially irreducible functional basis of elasticity tensors}\label{func-basis}
Based on Theorem \ref{Main_Thm} and \eqref{Sec_Ten}, we totally obtain 429 isotropic invariants. Notice that the elements of  intermediate tensors $\D^{(1)},\ \D^{(2)},\ \B,\ \C,\ \D,\ \F,\ \G,\ \HH,\ \K,\ \M,\ \N$ are related, there should be many polynomial relations among the invariants and their products. We utilize {\sl LinearSolve} function in Mathematica to seek all the polynomial relations and eliminate those invariants which are zeros or polynomials in others (refer to the supporting material for details). Finally there remains  251 isotropic invariants and they form a polynomially irreducible functional basis of $\E$. However, it can be quite tough finding the hidden functional relations. At least, we can draw a conclusion as following.
\begin{Theorem}
A set of 251 isotropic invariants forms a polynomially irreducible functional basis for elasticity tensor $\E$. These invariants are presented in Table \ref{IsoInv}.
\end{Theorem}

{\small
\begin{longtable}{|c|l|c|}
\caption{A polynomially irreducible functional basis of piezoelectric tensors.}\label{IsoInv}\\
  \hline
  Degree & Invariants & Number \\
  \hline
  1&$\lambda,\quad\mu,$&2\\
  \hline
  2&$J_2:=A_{ijkl}A_{ijkl},\quad\tr(\D^{(1)})^2,\quad\tr(\D^{(2)})^2,\quad\tr\D^{(1)}\D^{(2)},$&4\\
  \hline
  3&$J_3:=P_{ijkl}A_{ijkl},\quad\tr\HH,\quad\tr\K,\quad\tr(\D^{(1)})^3,\quad\tr(\D^{(2)})^3,\quad\tr\D^{(1)}\F,$&10\\
    &$\tr\D^{(1)}\G,\quad\tr\D^{(2)}\G,\quad\tr(\D^{(1)})^2\D^{(2)},\quad\tr\D^{(1)}(\D^{(2)})^2$,&\\
  \hline
  4&$J_{4}:=B_{ij}B_{ij},\quad\tr\F^2,\quad\tr\G^2,\quad\tr\D^{(1)}\C,\quad\tr\D^{(1)}\HH,\quad\tr\D^{(1)}\K,\quad\tr\D^{(1)}\M,$&16\\
  &$\tr\D^{(1)}\N,\quad\tr\D^{(2)}\C,\quad\tr\D^{(2)}\K,\quad\tr\D^{(2)}\M,\quad\tr\D^{(2)}\N,\quad\tr\F\G,\quad$&\\
 &$\tr(\D^{(1)})^2(\D^{(2)})^2,\quad\tr\D^{(1)}\D^{(2)}\F,\quad\tr\D^{(1)}\D^{(2)}\G,\quad$&\\
\hline
5&$J_5:=B_{ij}A_{ijkl}B_{kl},\quad\tr\B\HH,\quad\tr\B\K,\quad\tr\B\M,\quad\tr\B\N,\quad\tr\F\C,\quad\tr\F\HH,\quad$&29\\
&$\tr\F\K,\quad\tr\F\M,\quad\tr\F\N,\quad\tr\G\C,\quad\tr\G\HH,\quad\tr\G\K,\quad\tr\G\M,\quad\tr\G\N,\quad$&\\
&$\tr(\D^{(1)})^2\K,\quad\tr(\D^{(1)})^2\M,\quad\tr(\D^{(1)})^2\N,\quad\tr(\D^{(2)})^2\HH,\quad\tr(\D^{(2)})^2\N,\quad$&\\
&$\tr\D^{(1)}\F^2,\quad\tr\D^{(1)}\G^2,\quad\tr\D^{(2)}\F^2,\quad\tr\D^{(2)}\G^2,\quad\tr\D^{(1)}\D^{(2)}\C,\quad$&\\
&$\tr\D^{(1)}\D^{(2)}\HH,\quad\tr\D^{(1)}\D^{(2)}\K,\quad\tr\D^{(1)}\D^{(2)}\M,\quad\tr\D^{(1)}\D^{(2)}\N,\quad$&\\
 \hline
 6&$J_6:=B_{ij}P_{ijkl}B_{kl},\quad\tr\HH^2,\quad\tr\K^2,\quad\tr\M^2,\quad\tr\N^2,\quad\tr\F^3,\quad\tr\G^3,\quad$&46\\
 &$\tr\D^{(1)}\D,\quad\tr\D^{(2)}\D,\quad\tr\C\HH,\quad \tr\C\K,\quad\tr\C\M,\quad\tr\C\N,\quad\tr\HH\K,\quad$  &\\
 &$\tr\HH\M,\quad\tr\HH\N,\quad\tr\K\M,\quad\tr\K\N,\quad\tr\M\N,\quad\tr\F^2\G,\quad\tr\B\F^2,\quad\tr\B\G^2,\quad$&\\
 &$\tr\F\G^2,\quad\tr(\D^{(1)})^2\F^2,\quad\tr(\D^{(1)})^2\G^2,\quad\tr(\D^{(2)})^2\F^2,\quad\tr(\D^{(2)})^2\G^2,\quad$&\\
 &$\tr\D^{(1)}\B\K,\quad\tr\D^{(1)}\B\M,\quad\tr\D^{(1)}\B\N,\quad\tr\D^{(1)}\F\K,\quad\tr\D^{(1)}\F\N,\quad$&\\
 &$\tr\D^{(1)}\G\C,\quad\tr\D^{(1)}\G\HH,\quad\tr\D^{(1)}\G\K,\quad\tr\D^{(1)}\G\M,\quad\tr\D^{(1)}\G\N,\quad$&\\
 &$\tr\D^{(2)}\B\HH,\quad\tr\D^{(2)}\B\M,\quad\tr\D^{(2)}\B\N,\quad\tr\D^{(2)}\F\HH,\quad\tr\D^{(2)}\F\K,\quad$&\\
 &$\tr\D^{(2)}\F\M,\quad\tr\D^{(2)}\F\N,\quad\tr\D^{(2)}\G\HH,\quad\tr\D^{(2)}\G\M,$&\\
 \hline
 7&$J_7:=B_{ij}^2A_{ijkl}B_{kl},\quad\tr\F\D,\quad\tr\G\D,\quad\tr(\D^{(1)})^2\D,\quad\tr(\D^{(2)})^2\D,\quad$&54\\
 &$\tr\F^2\C,\quad\tr\F^2\HH,\quad\tr\F^2\K,\quad\tr\F^2\M,\quad\tr\F^2\N,\quad\tr\G^2\C,\quad\tr\G^2\HH,\quad$&\\
 &$\tr\G^2\K,\quad\tr\G^2\M,\quad\tr\G^2\N,\quad\tr\D^{(1)}\HH^2,\quad\tr\D^{(1)}\K^2,\quad\tr\D^{(1)}\M^2,\quad$&\\
 &$\tr\D^{(1)}\N^2,\quad\tr\D^{(2)}\HH^2,\quad\tr\D^{(2)}\K^2,\quad\tr\D^{(2)}\M^2,\quad\tr\D^{(2)}\N^2,\quad\tr\D^{(1)}\D^{(2)}\D,$&\\
 &$\tr\D^{(1)}\C\HH,\quad\tr\D^{(1)}\C\K,\quad\tr\D^{(1)}\C\M,\quad\tr\D^{(1)}\C\N,\quad\tr\D^{(1)}\HH\K,\quad$&\\
 &$\tr\D^{(1)}\HH\M,\quad\tr\D^{(1)}\HH\N,\quad\tr\D^{(1)}\K\M,\quad\tr\D^{(1)}\K\N,\quad\tr\D^{(1)}\M\N,\quad$&\\
 &$\tr\D^{(2)}\C\HH,\quad\tr\D^{(2)}\C\K,\quad\tr\D^{(2)}\C\M,\quad\tr\D^{(2)}\C\N,\quad\tr\D^{(2)}\HH\K,\quad$&\\
 &$\tr\D^{(2)}\HH\M,\quad\tr\D^{(2)}\HH\N,\quad\tr\D^{(2)}\K\M,\quad\tr\D^{(2)}\K\N,\quad\tr\D^{(2)}\M\N,\quad$&\\
 &$\tr\B\F\C,\quad\tr\B\F\K,\quad\tr\B\F\N,\quad\tr\B\G\C,\quad\tr\B\G\HH,\quad\tr\B\G\M,\quad$&\\
 &$\tr\F\G\HH,\quad\tr\F\G\K,\quad\tr\F\G\M,\quad\tr\F\G\N,\quad$&\\
 \hline
 8&$J_{8}:=B_{ij}^2P_{ijkl}B_{kl},\quad\tr\HH\D,\quad\tr\K\D,\quad\tr\M\D,\quad\tr\N\D,\quad\tr\B\HH^2,\quad\tr\B\K^2,\quad$&49\\
 &$\tr\B\M^2,\quad\tr\B\N^2,\quad\tr\F\C^2,\quad\tr\F\HH^2,\quad\tr\F\K^2,\quad\tr\F\M^2,\quad\tr\F\N^2,\quad\tr\G\C^2,$&\\
 &$\tr\G\HH^2,\quad\tr\G\K^2,\quad\tr\G\M^2,\quad\tr\G\N^2,\quad\tr\B^2\F^2,\quad\tr\B^2\G^2,\quad\tr\F^2\G^2,\quad$&\\
 &$\tr(\D^{(1)})^2\HH^2,\quad\tr(\D^{(1)})^2\K^2,\quad\tr(\D^{(1)})^2\N^2,\quad\tr(\D^{(2)})^2\HH^2,\quad\tr(\D^{(2)})^2\K^2,\quad$&\\
 &$\tr(\D^{(2)})^2\M^2,\quad\tr\D^{(1)}\G\D,\quad\tr\D^{(2)}\F\D,\quad\tr\B\HH\K,\quad$&\\
 &$\tr\B\HH\M,\quad\tr\B\HH\N,\quad\tr\B\K\M,\quad\tr\B\K\N,\quad\tr\B\M\N,\quad\tr\F\C\K,\quad\tr\F\C\N,\quad$&\\
 &$\tr\F\HH\K,\quad\tr\F\HH\N,\quad\tr\F\K\M,\quad\tr\F\K\N,\quad\tr\F\M\N,\quad\tr\G\C\HH,\quad\tr\G\C\M,\quad$&\\
 &$\tr\G\HH\K,\quad\tr\G\HH\M,\quad\tr\G\HH\N,\quad\tr\G\M\N,\quad$&\\
\hline
 9&$J_9:=B_{ij}^2A_{ijkl}B_{kl}^2,\quad\tr\HH^3,\quad\tr\K^3,\quad\tr\F^2\D,\quad\tr\G^2\D,\quad\tr\C^2\HH,\quad\tr\C^2\K,\quad$&29\\
 &$\tr\C^2\M,\quad\tr\C^2\N,\quad\tr\HH^2\K,\quad\tr\HH^2\N,\quad\tr\K^2\M,\quad\tr\M^2\N,\quad\tr\HH\K^2,\quad$&\\
 &$\tr\HH\N^2,\quad\tr\K\M^2,\quad\tr\M\N^2,\quad\tr\D^{(1)}\C\D,\quad\tr\D^{(1)}\K\D,\quad\tr\D^{(1)}\N\D,\quad$&\\
 &$\tr\D^{(2)}\C\D,\quad\tr\D^{(2)}\HH\D,\quad\tr\D^{(2)}\M\D,\quad\tr\F\G\D,\quad\tr\C\M\N,\quad\tr\HH\K\M,\quad$&\\
 &$\tr\HH\K\N,\quad\tr\HH\M\N,\quad\tr\K\M\N,\quad$&\\
 \hline
 10&$J_{10}:=B_{ij}^2P_{ijkl}B_{kl}^2,\quad\tr\B^2\HH^2,\quad\tr\B^2\K^2,\quad\tr\F^2\K^2,\quad$&10\\
 &$\tr\B\HH\K,\quad\tr\B\K\D,\quad\tr\F\K\D,\quad\tr\F\N\D,\quad\tr\G\HH\D,\quad\tr\G\M\D,$&\\
 \hline
 11&$\tr\D^{(1)}\D^2,\quad\tr\D^{(2)}\D^2.$&2\\
 \hline

 Total&&251\\
 \hline

  \end{longtable}

Due to the essential difference between our method and that of Olive, Kolev and Auffray (in algebraic viewpoint), our functional basis is not a proper subset of the Olive-Kolve-Auffray minimal integrity basis even if we need less isotropic invariants in total. For instance, in degrees 9 and 10, we have 29 and 10 isotropic invariants while the results of Olive, Kolev and Auffray are 21 and 7 respectively. In particular, to make an explicit comparison, we list the numbers of isotropic invariants in each degree of our polynomially irreducible functional basis (PIFB) and Olive-Kolev-Auffray minimal integrity basis (MIB) in Table \ref{Num}.
\begin{table}[htbp]
\caption{Numbers of isotropic invariants in different degrees.}\label{Num}
\centering
\begin{tabular}{|c|c|c||c|c|c|}
\hline
Degree & PIFB & MIB  &Degree & PIFB & MIB  \\
\hline
1 & 2 & 2 & 7 & 54 & 76\\\hline
2 & 4 & 4 & 8 & 49 & 66 \\\hline
3 & 10 & 10 & 9 & 29 & 21\\\hline
4 & 16 & 16 & 10 & 10 & 7\\\hline
5 & 29 & 33 & 11 & 2 & 5\\\hline
6 & 46 & 57 & Total & 251 & 297\\\hline

\end{tabular}
\end{table}

\section{Conclusions}
We extend Smith's approach and propose a possible methodology for constructing functional basis of high order tensor by designing a series of intermediate tensors with orders no greater than two to determine the orbit of the considered tensor. After a careful discussion and examination, we finally obtain an polynomially irreducible functional basis consisting of 251 isotropic invariants for elasticity tensors. Moreover, we should note that our functional basis is not necessary a minimal functional basis, but it provides a smaller upper bound for this problem.

\end{document}